\documentclass{article}
\usepackage{amsfonts,amsmath,amsthm,amssymb,authblk,latexsym,bm,color,eucal}
\usepackage{algorithm}
\usepackage{algpseudocode}
\usepackage{xcolor}  
\newtheorem{Theorem}{Theorem}[section]
\newtheorem{lem}[Theorem]{Lemma}
\newtheorem{Remark}[Theorem]{Remark}
\newtheorem{Definition}[Theorem]{Definition}

\newtheorem{Example}[Theorem]{Example}
\numberwithin{equation}{section}
\usepackage[latin1]{inputenc}
\begin{document}
\title{{\LARGE Decoding Algorithms for Twisted GRS Codes}}

\author{Guanghui~Zhang$^1$, Liren Lin$^2$, Bocong Chen$^3$\footnote{E-mail addresses: {\it zghui@squ.edu.cn (G. Zhang); l\underline{~}r\underline{~}lin86@163.com (L. Lin);
bocongchen@foxmail.com (B. Chen)}}
}

\date{\small
$1.$ School of Mathematics and Physics,
Suqian University,
Suqian, Jiangsu 223800, China\\
$2.$ Heibei Key Laboratory of Applied Mathematics, Faculty of Mathematics and Statistics,
Hubei University,
Wuhan, Hubei 430062, China\\
$3.$ School of Mathematics, South China University of Technology, Guangzhou 510641, China\\
}


\maketitle

\begin{abstract}
Twisted generalized Reed-Solomon (TGRS) codes
were introduced to extend the algebraic capabilities
of classical generalized Reed-Solomon (GRS) codes.
This extension holds the potential for constructing
new non-GRS maximum distance separable (MDS) codes
and enhancing cryptographic security.
It is known that TGRS codes with $1$ twist can either be MDS or near-MDS.
In this paper, we employ the Gaussian elimination method
to propose new decoding algorithms for
MDS  TGRS codes with parameters $[n,k,n-k+1]$.
The algorithms can correct
up to  $\lfloor \frac{n-k}{2}\rfloor$ errors when $n-k$ is odd,
 and
$\lfloor \frac{n-k}{2}\rfloor-1$ errors
 when $n-k$ is even.
The computational complexity for both scenarios is $O(n^3)$. 
Our approach diverges from existing methods based on Euclidean
algorithm and addresses situations that have not been considered in the existing literature \cite{SYJL}. Furthermore, this method is also applicable to decoding near-MDS TGRS codes with parameters $[n, k, n-k]$,
enabling correction of up to $\lfloor \frac{n-k-1}{2} \rfloor$ errors, while maintaining polynomial time complexity in $n$.



\end{abstract}

\section{Introduction}
The process of decoding--determining which codeword
(and thus which message $\mathbf{c}$) was sent when a vector
$\mathbf{y}$ is received--is complex.
It is not merely the final step in a communication or storage system,
but is where the entire purpose of coding is realized.
Therefore, finding efficient (fast) decoding algorithms is a major area of research in coding theory due to their practical applications,
see \cite{Ball, CXY, DB, LAH, LTX2021, LTX2023, LXY, TH, Wu, YLHL} and references therein.
Coding theory designs codes that add controlled redundancy to detect and/or correct errors caused by noisy channels,
defective memory cells, or malicious interference.

The error correction capability of a code is closely
related to its minimum Hamming distance.
If $d$  is the minimum Hamming distance of a code
(whether linear or nonlinear), then the code can correct up to $\lfloor\frac{d-1}{2}\rfloor$ errors.
Here, as usual, for a real number  $x\in \mathbb{R}$,
$\lceil x\rceil$ denotes the ceiling function,
which outputs the smallest integer greater than or equal to $x$.
Additionally,
 $\lfloor x\rfloor$ represents the floor function,
 outputting the largest integer less than or equal to $x$.

Let $\mathbb{F}_q$ be a finite field with $q$ elements.
A linear code $\mathcal{C}$ over $\mathbb{F}_q$ with
parameters $[n,k,d]$ is called a \emph{maximum distance
separable (MDS) code} if it meets the Singleton bound, i.e.,
$d=n-k+1$.
MDS codes guarantee the maximum possible minimum distance
for a given $[n, k]$ linear code,
enabling recovery of the original $k$ symbols from any $n-k$ erasures
or $\lfloor \frac{n-k}{2}\rfloor$ errors.
Hence, efficient decoders facilitate system designers
in achieving Shannon-type reliability bounds
without unnecessary redundancy.
A linear code $\mathcal{C}$ over $\mathbb{F}_q$ with parameters $[n,k,d]$
 is termed an  \emph{almost MDS (AMDS)} code if $d=n-k$.
Furthermore, a
linear code  $\mathcal{C}$  over $\mathbb{F}_q$
is called a \emph{near MDS (NMDS)} code if both $\mathcal{C}$
and the dual code of $\mathcal{C}$ are AMDS codes.
NMDS codes sacrifice at most one symbol of
minimum distance compared with their MDS counterparts.
In practice this means the system tolerates either one additional erasure or
two additional errors per codeword compared with an MDS code of the same length and rate.
It is often an acceptable price when decoding speed is at a premium.

Twisted generalized Reed-Solomon (TGRS) codes
are a significant extension of classical Reed-Solomon (RS)
and generalized Reed-Solomon (GRS) codes \cite{BPN}.
They incorporate additional algebraic structures to
enhance flexibility.
These codes find applications in areas such as deep-space communication, data storage (e.g., SSDs, CDs), and cryptographic
systems where reliable data recovery is crucial \cite{HLL, ZLL, ZWC}.
Because GRS codes possess an explicit algebraic structure, they are inherently vulnerable to Sidelnikov-Shestakov attacks. Introducing a twist breaks this structure and thus masks the code¡¯s origin, while still preserving the advantageous error-correcting properties of the underlying algebraic construction.
Recently, TGRS codes and their subfield subcodes, such as twisted Goppa codes, have garnered significant attention due to their promising applications in coding theory and post-quantum cryptography \cite{BBP, GZ, HY, LL, LLO, SYLH, SZS, ZZT}. They generalize classical RS  and Goppa codes by introducing a ``twist" term in the polynomial evaluation structure, offering richer algebraic properties and enhanced structural flexibility.
It is known that TGRS codes with $1$ twist can either be MDS or near-MDS.

In this paper, we focus on  TGRS codes
and present   decoding algorithms for this class of codes.
Effective methods for decoding GRS codes have been
discussed in \cite{Roth, SKHN},
including the Peterson-Gorenstein-Zierler Decoding Algorithm, Berlekamp-Massey Decoding Algorithm, and Sugiyama Decoding Algorithm.
Notably, the Sugiyama Decoding Algorithm employs the Euclidean Algorithm for polynomials in a distinctive yet powerful manner.
Recently, Sun et al. \cite{SYJL}  proposed new decoding algorithms using the extended Euclidean algorithm for two classes of MDS TGRS codes with parameters $[n,k,n-k+1]$  when $n-k$ is even.
Their decoding algorithms can correct
 up to $\lfloor \frac{n-k}{2}\rfloor$ errors with time complexity
 of  $O(qn)$,
 and they are applicable to both TGRS codes and twisted Goppa codes. Their key contribution lies in an enhanced decoding method that achieves
 $\frac{n-k}{2}$-error correction for even-degree Goppa polynomials,
 thus improving upon previous bounds.

Building on these foundational works,
this paper further investigates decoding algorithms for TGRS codes.
Specifically, we present unified decoding algorithms that accommodate both MDS and NMDS TGRS codes with flexible parameters, extending previous results that were restricted to even differences between code length and code dimension.
Our findings encompass certain situations discussed in \cite{SYJL}.
Moreover, the hook $h$ of each   TGRS code in this paper is arbitrary,
contrasting with the restricted cases of $h = 0$ or $k-1$ in \cite{SYJL}.
Our approach utilizes Gaussian elimination to efficiently solve twisted polynomials, achieving polynomial-time decoding complexity $O(n^3)$.
Consequently, in certain scenarios,
the algorithms provided in this paper can demonstrate greater efficiency.
It is also worth mentioning that our algorithms apply not only to
 TGRS codes but also to twisted Goppa codes,
 thereby broadening their applicability compared to existing methods.

The remainder of this paper is organized as follows:
In Section 2, we introduce basic notation and results
concerning MDS, AMDS, NMDS, and TGRS codes.
In Sections 3 and 4, we provide decoding algorithms for MDS and NMDS TGRS codes, respectively, based on the Gaussian elimination method.
Section 5 offers a comparison with existing results.
Finally, Section 6 concludes the paper and discusses potential future work.

\section{Preliminaries}
Let $\mathbb{F}_q$ be the finite field with $q$ elements,
where $q$ is a prime power.
Let $n$ be a positive integer and
let $\mathbb{F}_q^n$ denote the vector space of
all $n$-tuples over the finite field $\mathbb{F}_q$.
We typically represent the vectors in
 $\mathbb{F}_q^n$ as row vectors.
Let $\mathcal{C}$ be a non-empty subset of $\mathbb{F}_q^n$.
If $\mathcal{C}$ forms a subspace of $\mathbb{F}_q^n$, we
call  $\mathcal{C}$
\emph{a linear code}. The vectors in $\mathcal{C}$
are referred to as \emph{codewords}.
If $\mathcal{C}$ has dimension $k$ over $\mathbb{F}_q$, we say that $\mathcal{C}$ is an $[n,k]$ linear code
over $\mathbb{F}_q$. A \emph{generator matrix} for an $[n,k]$ linear code $\mathcal{C}$ is any $k\times n$ matrix $G$ whose rows form a
basis for $\mathcal{C}$.
For any vectors
$\mathbf{a}=(a_1,a_2,\cdots,a_n)$
and
$\mathbf{b}=(b_1,b_2,\cdots,b_n)$ in $\mathbb{F}_q^n$,
the inner product $\mathbf{a}\cdot \mathbf{b}$ is defined as
$\mathbf{a}\cdot \mathbf{b}=\sum_{i=1}^na_ib_i$.
Let $\mathcal{C}$ be a $k$-dimensional linear code over $\mathbb{F}_q$.
The \emph{dual code} of  $\mathcal{C}$ is defined as
$$ \mathcal{C}^\perp=\big\{\mathbf{x}\in \mathbb{F}_q^n
\,\big|\,\mathbf{x}\cdot \mathbf{c}=0 ~\hbox{for any}~\mathbf{c}\in \mathcal{C}\big\}.$$
A generator matrix of the dual code   $\mathcal{C}^\perp$
is referred to as a \emph{parity-check matrix} for $\mathcal{C}$.
Thus, if  $H$ is  a  parity-check matrix  for $\mathcal{C}$, we can express
$\mathcal{C}$ as
$$\mathcal{C}=\big\{\mathbf{c}\in \mathbb{F}_q^n\,|\,H\mathbf{c}^T=\mathbf{0}^T\},$$
where $\mathbf{c}^T$ denotes  the transpose of the vector $\mathbf{c}$.


An important parameter of a linear code is
its minimum Hamming distance.
The Hamming distance $d(\mathbf{x},\mathbf{y})$
between two vectors $\mathbf{x},\mathbf{y}\in \mathbb{F}_q^n$
is defined to be the number
of coordinates in which $\mathbf{x}$ and $\mathbf{y}$ differ.
The \emph{minimum Hamming distance} of a linear code $\mathcal{C}$
is the smallest Hamming distance between any two distinct codewords in
$\mathcal{C}$,
and
 it plays a crucial role in determining the
 error-correcting capability of $\mathcal{C}$.
If the minimum Hamming distance $d$ of an $[n,k]$ linear code is known,
 we refer to the  code as an $[n,k,d]$ code.
An $[n,k,d]$ linear code
can correct up to $\lfloor\frac{d-1}{2}\rfloor$ errors.

If an $[n,k,d]$ linear code over $\mathbb{F}_q$ exists,
 then it satisfies the Singleton bound:
$$d\leq n-k+1.$$
An $[n,k,d]$ linear code for which equality holds in the
Singleton Bound is called \emph{maximum distance separable} (MDS).
An $[n,k,d]$ linear code $\mathcal{C}$ over $\mathbb{F}_q$
is termed \emph{almost MDS (AMDS)} if $d=n-k$.
Furthermore, an $[n,k,d]$ linear code $\mathcal{C}$
is classified as \emph{near MDS (NMDS)} if both
$\mathcal{C}$ and its dual code $\mathcal{C}^\perp$ are AMDS codes.

Now, let us recall
the definition of twisted generalized Reed-Solomon codes
(see \cite{BBP, BPN, LR}).
To facilitate understanding,
we  first restate the notion of twisted polynomials.

\begin{Definition}\label{def1}
Let $k,t$ and $h$ be positive integers such that
$0\leq h<k\leq q$,
and let $\eta\in \mathbb{F}^*_q=\mathbb{F}_q\backslash\{0\}$.
We define the set of $(k,t,h,\eta)$-twisted polynomials by
$$\mathcal{V}_{(k,t,h,\eta)}=\bigg\{f(x)=\sum_{i=0}^{k-1}a_ix^i+\eta a_hx^{k-1+t}\,\bigg|\, a_i\in \mathbb{F}_q~\hbox{for}~0\leq i\leq k-1 \bigg\}.$$
\end{Definition}

We are now ready to present the definition
of twisted generalized Reed-Solomon codes.
\begin{Definition}\label{def2}
Let $\alpha_1,\alpha_2,\cdots,\alpha_n$ be distinct
elements of  $\mathbb{F}_q$ and
write $\alpha=(\alpha_1,\alpha_2,\cdots,\alpha_n)$.
Let $v_1,v_2,\cdots,v_n$ be nonzero elements in $\mathbb{F}_q$ and
define $\mathbf{v}=(v_1,v_2,\cdots,v_n)$.
Consider $k,t,h,\eta$
 chosen as in Definition \ref{def1} such that $k<n$ and $t\leq n-k$.
We keep the notation $\mathcal{V}_{(k,t,h,\eta)}$ as
used  in Definition \ref{def1}.
The twisted generalized Reed-Solomon code of
length $n$ and dimension $k$ is defined as
$$\mathcal{C}_{k,t,h}(\alpha,\mathbf{v},\eta)=
\Big\{\big(v_1f(\alpha_1),v_2f(\alpha_2),\cdots,v_nf(\alpha_n)\big)
\,\Big|\,f(x)\in \mathcal{V}_{(k,t,h,\eta)}\Big\},$$
where  $h$ is referred to as the hook and  $t$ as the twist.
\end{Definition}
For convenience, we abbreviate twisted generalized
Reed-Solomon codes as TGRS codes.
In this paper, we focus on studying the decoding algorithm for
a class of TGRS codes.
By virtue of the equivalence of codes
(refer to \cite[Section 2.1]{Mac}),  we
may assume that $\mathbf{v}=\mathbf{1}=(1,1,\cdots,1)$.
Additionally, we always assume that $t=1$.
 Hence, our aim is to present a decoding
 algorithm for the TGRS codes
 $\mathcal{C}_{k,1,h}(\alpha,\mathbf{1},\eta)$,
 where $\mathbf{1}=(1,1,\cdots,1)$.
The TGRS code $\mathcal{C}_{k,1,h}(\alpha,\mathbf{1},\eta)$ is
given by
$$\mathcal{C}_{k,1,h}(\alpha,\mathbf{1},\eta)=
\Big\{\big(f(\alpha_1),f(\alpha_2),\cdots,f(\alpha_n)\big)
\,\Big|\,f(x)\in \mathcal{V}_{(k,1,h,\eta)}\Big\},$$
where $f(x)\in \mathcal{V}_{(k,1,h,\eta)}$
means that $f(x)$ takes the form
$$f(x)=\sum_{i=0}^{k-1}a_ix^i+\eta a_hx^{k}.$$
Thus, this type of TGRS code has the following
generator matrix:
$$
\begin{pmatrix}
1 & 1 & \cdots & 1\\
\alpha_1 & \alpha_2 & \cdots & \alpha_n\\
\vdots & \vdots &\ddots & \vdots\\
\alpha_1^{h-1} & \alpha_2^{h-1} & \cdots & \alpha_n^{h-1}\\
\alpha_1^h+\eta\alpha_1^k & \alpha_2^h+\eta\alpha_2^k & \cdots & \alpha_n^h+\eta\alpha_n^k\\
\alpha_1^{h+1} & \alpha_2^{h+1} & \cdots & \alpha_n^{h+1}\\
\vdots & \vdots &\ddots & \vdots\\
\alpha_1^{k-1} & \alpha_2^{k-1} & \cdots & \alpha_n^{k-1}
\end{pmatrix}.
$$

\section{A decoding algorithm for MDS TGRS codes}
In this section, we investigate the decoding algorithm for
MDS
TGRS codes.
We begin by presenting the necessary and sufficient
condition under which a TGRS code qualifies as
an MDS code (see \cite{BPN, HY}).
Following this, we will provide a decoding algorithm
tailored for this class of MDS codes.
The following lemma will be instrumental in
determining whether a TGRS code is MDS:
\begin{lem}(\cite[Lemma 2.6]{HY})\label{mdscondition}
Keeping the notation
$\mathcal{C}_{k,1,h}(\alpha,\mathbf{1},\eta)$ as defined previously,
 the TGRS code $\mathcal{C}_{k,1,h}(\alpha,\mathbf{1},\eta)$ is MDS if and only if
$$
\eta\sum_{i\in I}\alpha_i\neq -1,
~\hbox{for any}~I\subset\{1,2,\cdots,n\}~\hbox{with}~|I|=k.
$$
\end{lem}
Now we focus on presenting a decoding algorithm for an $[n, k]$
 MDS  TGRS  code where $n-k$ is odd,
 utilizing the Gaussian elimination method.
 The following theorem establishes that there exists a decoding algorithm capable of correcting up to $\lfloor \frac{d-1}{2} \rfloor$ errors,
 where $d = n - k + 1$ is the minimum distance.
\begin{Theorem}\label{MDSGRSisodd}
Let $\mathcal{C}_{k,1,h}(\alpha,\mathbf{1},\eta)$ be defined as above.
Assume that $n-k$ is odd. Then there exists a decoding algorithm for
 the $k$-dimensional MDS TGRS code
 $\mathcal{C}_{k,1,h}(\alpha,\mathbf{1},\eta)$ of length $n$,
which corrects up to $\lfloor \frac{n-k}{2}\rfloor$ errors and completes in a number of operations which is polynomial in $n$.
\end{Theorem}
\begin{proof}
Suppose  we have received the vector
$(y_1,y_2,\cdots,y_n)$ in $\mathbb{F}_q^n$.
Our goal is to find the
$(k,1,h,\eta)$-twisted polynomial
$f(x)\in \mathcal{V}_{(k,1,h,\eta)}$ of degree at most $k$
such that
\begin{equation}\label{error vector31}
(y_1,y_2,\cdots,y_n)=(f(\alpha_1),f(\alpha_2),\cdots,f(\alpha_n))+\mathbf{e},
\end{equation}
where $\mathbf{e}$ is the error vector with weight at most $\lfloor \frac{n-k}{2}\rfloor$.
Since $\mathcal{C}_{k,1,h}(\alpha,\mathbf{1},\eta)$ is an MDS code,
we can state that
$$\frac{n-k}{2}=\frac{d-1}{2}.$$

Let $h(x)$ be an arbitrary polynomial of degree at most $\lfloor\frac{n-k}{2}\rfloor$ and let $g(x)$ be an arbitrary polynomial of degree at most $k+\lceil\frac{n-k}{2}\rceil-1$.
We will find
the coefficients of $g(x)$ and $h(x)$ by solving the system of $n$ equations:
$$g(\alpha_j)-h(\alpha_j)y_j=0,~\mbox{for} ~j=1,2,\cdots,n.$$
This homogeneous linear system has
$$\deg(h(x))+1+\deg(g(x))+1
=\lfloor\frac{n-k}{2}\rfloor+1+k+\lceil\frac{n-k}{2}\rceil-1+1=n+1$$
unknowns (the coefficients of $g(x)$ and $h(x)$)
and $n$ equations. Hence, we can find a non-trival solution for $h(x)$ and
$g(x)$
using Gaussian elimination in polynomial time with respect to
$n$. It is readily seen that  $h(x)\neq 0$.

By the assumption (\ref{error vector31}),
there is a $(k,1,h,\eta)$-twisted polynomial $f(x)$ of
degree at most $k$
such that $y_j=f(\alpha_j)$ for at least
 $n-\lfloor\frac{n-k}{2}\rfloor$ values of $j$.
Let
$$
J=\big\{j\,\big|\,y_j=f(\alpha_j)~\hbox{for}~1\leq j\leq n\big\}.
$$
Then we have the lower bound for the number of elements in the set $J$:
$$|J|\geq n-\lfloor\frac{n-k}{2}\rfloor.$$
Furthermore, we can observe that
$$g(\alpha_j)-h(\alpha_j)f(\alpha_j)=g(\alpha_j)-h(\alpha_j)y_j=0
~\hbox{for any}~ j\in J.$$
Thus, for these values of $j\in J$,
$\alpha_j$ is a zero of $g(x)-f(x)h(x)$.
Consequently, the polynomial $g(x)-f(x)h(x)$ has at least $|J|\geq n-\lfloor\frac{n-k}{2}\rfloor$ distinct zeros in $\mathbb{F}_q$.
Next, we need to consider the degree of the polynomial
$\varphi(x)=g(x)-h(x)f(x)$. Suppose  $\varphi(x)\neq 0$. It follows that
\begin{eqnarray*}
\deg(\varphi(x))&=& \max\big\{\deg(g(x)), \deg(h(x)f(x))\big\}\\
&\leq & \max\big\{k+\lceil\frac{n-k}{2}\rceil-1, \lfloor\frac{n-k}{2}\rfloor+k\big\}.\\
\end{eqnarray*}
Since $n-k$ is odd, we simplify:
\begin{eqnarray*}
k+\lceil\frac{n-k}{2}\rceil-1&=&k+(n-k-\lfloor\frac{n-k}{2}\rfloor)-1\\
&=&n-\lfloor\frac{n-k}{2}\rfloor-1\\
&=& n-\frac{n-k-1}{2}-1\\
&=&\frac{n+k-1}{2}\\
&=&\frac{n-k-1}{2}+k\\
&=&\lfloor\frac{n-k}{2}\rfloor+k,
\end{eqnarray*}
yielding
$$\deg(\varphi(x))=\deg(g(x)-h(x)f(x))\leq\deg(g(x))=k+\lceil\frac{n-k}{2}\rceil-1.$$
Since the number of zeros of a polynomial cannot exceed its degree, we therefore have
$$n-\lfloor\frac{n-k}{2}\rfloor\leq k+\lceil\frac{n-k}{2}\rceil-1.$$
However, we can re-examine this relationship:
\begin{eqnarray*}
k+\lceil\frac{n-k}{2}\rceil-1& = &k+n-k-\lfloor\frac{n-k}{2}\rfloor-1 \\
&=& \frac{n+k-1}{2}\\
&<&\frac{n+k+1}{2}\\
&=&n-\lfloor\frac{n-k}{2}\rfloor.
\end{eqnarray*}
This leads to a contradiction, implying that $\varphi(x)$
 must be identically zero. Therefore, $h(x)$ divides $g(x)$ and the quotient is $f(x)$.
\end{proof}

For the case where $n-k$ is even, our method remains applicable;
however, the decoding algorithm can only correct up to
$\lfloor \frac{d-1}{2}\rfloor-1=\lfloor \frac{n-k}{2}\rfloor-1$ errors.
Although the proof follows similarly to that of Theorem \ref{MDSGRSisodd},
for completeness it is presented along with the theorem.
\begin{Theorem}\label{MDSGRSiseven}
Let $\mathcal{C}_{k,1,h}(\alpha,\mathbf{1},\eta)$ be defined as above.
Assume that $n-k$ is even. Then,
 there exists a decoding algorithm for the $k$-dimensional MDS TGRS code $\mathcal{C}_{k,1,h}(\alpha,\mathbf{1},\eta)$ of length $n$,
which corrects up to $\frac{n-k}{2}-1$ errors and completes in a number of operations which is polynomial in $n$.
\end{Theorem}

\begin{proof}
Suppose that we have received the vector $(y_1,y_2,\cdots,y_n)$. We want to find the $(k,1,h,\eta)$-twisted polynomial $f(x)\in \mathcal{V}_{(k,1,h,\eta)}$ of degree at most $k$
such that
\begin{equation}\label{error vector32}
(y_1,y_2,\cdots,y_n)=(f(\alpha_1),f(\alpha_2),\cdots,f(\alpha_n))+\mathbf{e},
\end{equation}
where $\mathbf{e}$ is the error vector of weight at most $\lfloor \frac{n-k}{2}\rfloor-1=\frac{n-k}{2}-1$.
Since $\mathcal{C}_{k,1,h}(\alpha,\mathbf{1},\eta)$ is an MDS code,
one has
$$\frac{n-k}{2}=\frac{d-1}{2}.$$

Let $h(x)$ be an arbitrary polynomial of degree at most $\lfloor\frac{n-k}{2}\rfloor=\frac{n-k}{2}$ and let $g(x)$ be an arbitrary polynomial of degree at most $k+\lceil\frac{n-k}{2}\rceil=\frac{n+k}{2}$.
We determine the coefficients of $g(x)$ and $h(x)$ by solving the system of $n$ equations:
$$g(\alpha_j)-h(\alpha_j)y_j=0,~\mbox{for} ~j=1,2,\cdots,n.$$
This homogeneous linear system has
$$\deg(h(x))+1+\deg(g(x))+1
=\lfloor\frac{n-k}{2}\rfloor+1+k+\lceil\frac{n-k}{2}\rceil+1=n+2$$
unknowns (the coefficients of $g(x)$ and $h(x)$)
and $n$ equations. Hence, we can find a non-trival solution for $h(x)$ and
$g(x)$) in a number of operations that is polynomial in $n$ using Guassian elimination. Clearly, $h(x)\neq 0$.

By assumption (\ref{error vector32}), there is a polynomial $f(x)$ of degree at most $k$
such that $y_j=f(\alpha_j)$ for at least $n-\lfloor\frac{n-k}{2}\rfloor+1$ values of $j$.
Thus we let $J=\{j\,|\,y_j=f(\alpha_j), 1\leq j\leq n\}$ and then the number $|J|$ of elements in the set $J$  has the low bound as follows:
$$|J|\geq n-\lfloor\frac{n-k}{2}\rfloor+1=\frac{n+k}{2}+1.$$

On the other hand, we have
$$g(\alpha_j)-h(\alpha_j)f(\alpha_j)=g(\alpha_j)-h(\alpha_j)y_j=0
 ~\hbox{for any}~j\in J.$$
That is to say, for these values of $j\in J$, $\alpha_j$ is a zero of $g(x)-f(x)h(x)$.
Hence, this polynomial $g(x)-f(x)h(x)$ has at least $|J|\geq n-\lfloor\frac{n-k}{2}\rfloor+1=\frac{n+k}{2}+1$ distinct zeros in $\mathbb{F}_q$.

Now we consider the degree of the polynomial $\varphi(x)=g(x)-h(x)f(x)$. Suppose that $\varphi(x)\neq 0$. It follows that
\begin{eqnarray*}
\deg(\varphi(x))&=& \max\big\{\deg(g(x)), \deg(h(x)f(x))\big\}\\
&\leq & \max\big\{\frac{n+k}{2}, \frac{n-k}{2}+k\big\}\\
&=&\frac{n+k}{2}.
\end{eqnarray*}
This gets a contradiction since $\varphi(x)$ has at least $\frac{n+k}{2}+1$ distinct zeros.
Thus,  $\varphi(x)$ is identically zero. Therefore, $h(x)$ divides $g(x)$ and the quotient is $f(x)$.
\end{proof}

In the following, we provide examples to illustrate the above results.

\begin{Example}{\rm
Assume that $\mathbb{F}_9=\mathbb{F}_3[z]/\langle z^2+z+2\rangle=\{0,1,2,z,z+1,z+2,2z,2z+1,2z+2\}$, where $z^2=2z+1$.
Let $n=5,k=2, \eta=z\in \mathbb{F}_9$ and $\alpha=(\alpha_1,\alpha_2,\alpha_3,\alpha_4,\alpha_5)=(0,1,z,z+1,2z)$.
Then, consider the sums:
$$
\Big\{\sum_{i\in I}\alpha_i\,\big|\,~\hbox{for any}~I\subset\{1,2,\cdots,n\}
~\hbox{with}~ |I|=2\Big\}=\Big\{0,1,z,z+1,2z,z+2,2z+1\Big\},
$$
which shows that
$$
\eta\sum_{i\in I}\alpha_i=z\sum_{i\in I}\alpha_i\neq -1~\hbox{for any} ~I\subset\{1,2,\cdots,n\}~\hbox{with}~|I|=2.
$$
Thus, by Lemma \ref{mdscondition},
we can conclude that the TGRS code
$\mathcal{C}_{2,1,0}(\alpha,\mathbf{1},z)$ is MDS,
where
$$
\mathcal{C}_{2,1,0}(\alpha,\mathbf{1},z)=
\big\{(f(0),f(1),f(z),f(z+1),f(2z))\,\big|\,
f(x)=a_0+a_1x+za_0x^2 ~\hbox{for any}~ a_0,a_1\in \mathbb{F}_9\big\}.
$$
We   sent a codeword $\mathbf{u}$ of the $2$-dimensional MDS GRS code $\mathcal{C}_{2,1,0}(\alpha,\mathbf{1},z)$ over $\mathbb{F}_9$,
 and
we have received
$$\mathbf{y}=(y_1,y_2,y_3,y_4,y_5)=(1,2z,2z+1,2z+1,2z+2)=\mathbf{u}+\mathbf{e},$$
where the weight of $\mathbf{e}$ is at most $1$.
According to the algorithm in the proof of Theorem \ref{MDSGRSisodd},
we need to find a polynomial $g(x)$ of degree at most $3$ and a polynomial $h(x)$ of degree at most $1$, such that
$$g(\alpha_j)=h(\alpha_j)y_j, $$
for $j=1,2,3,4,5$.
The resulting equations are:
$$
\begin{cases}
g(0)=h(0),\\
g(1)=2zh(1),\\
g(z)=(2z+1)h(z),\\
g(z+1)=(2z+1)h(z+1),\\
g(2z)=(2z+2)h(2z).
\end{cases}
$$
Let us assume that
$$g(x)=d_3x^3+d_2x^2+d_1x+d_0$$
and
$$h(x)=c_1x+c_0,$$
for some coefficients $d_3,d_2, d_1, d_0, c_1, c_0\in \mathbb{F}_9$.
From the equations, we can form a linear system represented as:
$$
\begin{pmatrix}
1 & 0 & 0 & 0 & 2 & 0\\
1 & 1 & 1 & 1 & z & z\\
1 & z & 2z+1 & 2z+2 & z+2 & z+1\\
1 & z+1 & z+2 & 2z & z+2 & 2z\\
1 & 2z & 2z+1 & z+1 & z+1 & 2
\end{pmatrix}
\begin{pmatrix}
d_0 \\ d_1\\d_2 \\ d_3 \\c_0\\ c_1
\end{pmatrix}=\mathbf{0}.
$$
This implies all solutions to the system can be expressed as
$$
\begin{pmatrix}
d_0 \\ d_1\\d_2 \\ d_3 \\c_0\\ c_1
\end{pmatrix}=\nu
\begin{pmatrix}
1 \\ z\\2 \\ 2z \\1\\ 2
\end{pmatrix}
$$
for some $\nu\in \mathbb{F}_9$.
Therefore we have
$$
\begin{cases}
g(x)=\nu(1+zx+2x^2+2zx^3),\\[2pt]
h(x)=\nu(1+2x).
\end{cases}
$$
We can validate that $g(1)=h(1)=0$, yeilding
\begin{eqnarray*}
f(x)=\frac{g(x)}{h(x)}&=&\frac{\nu(1+2x)([1+(1+z)x+zx^2]}{\nu(1+2x)}\\
&=&1+(1+z)x+zx^2\\
&=&a_0+a_1x+za_0x^2.
\end{eqnarray*}
Evaluating the polynomial $f(x)$, we deduce that
$$\mathbf{u}=\big(f(0),f(1),f(z),f(z+1),f(2z)\big)=(1,2z+2,2z+1,2z+1,2z+2).$$
}
\end{Example}


\begin{Example}{\rm
Assume that $\mathbb{F}_{16}=\mathbb{F}_2[z]/\langle z^4+z+1\rangle=\{a+bz+cz^2+dz^3|a,b,c,d\in \mathbb{F}_2\}$, where $z^4=z+1$.
Let $n=8,k=2, \eta=z^2\in \mathbb{F}_{16}$ and
$$\alpha=(\alpha_1,\alpha_2,\alpha_3,\alpha_4,\alpha_5,\alpha_6,\alpha_7,\alpha_8)=(0,1,z,z+1,z^2,z^2+1,z^2+z,z^2+z+1).$$
Simple algebraic calculations show
$$
\Big\{\sum_{i\in I}\alpha_i\,\big|\,~\hbox{for any} ~I\subset\{1,2,\cdots,n\}
~\hbox{with}~|I|=2\Big\}=\Big\{0,1,z,z+1,z^2,z^2+1,z^2+z,z^2+z+1\Big\},
$$
giving
$$
\eta\sum_{i\in I}\alpha_i=z^2\sum_{i\in I}\alpha_i\neq -1=1~\hbox{for any} ~I\subset\{1,2,\cdots,n\}~\hbox{with}~|I|=2.
$$
Thus, by Lemma \ref{mdscondition}, we conclude that the TGRS code $\mathcal{C}_{2,1,0}(\alpha,\mathbf{1},z^2)$ is MDS,
where $\mathcal{C}_{2,1,0}(\alpha,\mathbf{1},z^2)$ is given by
\begin{equation*}
\begin{split}
\big\{(f(0),f(1),f(z),f(z+1),f(z^2),f(z^2+1),f(z^2+z),f(z^2+z+1))
&\,\big|\,\\
f(x)=a_0+a_1x+z^2a_0x^2 ~\hbox{for any}~a_0,a_1\in \mathbb{F}_{16}\big\}.
\end{split}
\end{equation*}
Now suppose we have sent a codeword $\mathbf{u}$ of
the $2$-dimensional MDS TGRS code $\mathcal{C}_{2,1,0}(\alpha,\mathbf{1},z^2)$ over $\mathbb{F}_{16}$ and
received 
$$\mathbf{y}=(y_1,y_2,y_3,y_4,y_5,y_6,y_7,y_8)=
(1,z^2+z,z^2+z,1,z^2+1,z+1,z,z^2)=\mathbf{u}+\mathbf{e},$$
where the weight of $\mathbf{e}$ is at most $1$.
Following the algorithm in the proof of Theorem \ref{MDSGRSiseven},
we want to find a polynomial $g(x)$ with degree at most $5$
and a polynomial $h(x)$ with degree at most $3$, such that
$$g(\alpha_j)=h(\alpha_j)y_j, $$
for $j=1,2,3,4,5,6,7,8$.
Assume
$$g(x)=d_5x^5+d_4x^4+d_3x^3+d_2x^2+d_1x+d_0$$
and
$$h(x)=c_3x^3+c_2x^2+c_1x+c_0,$$
for some $d_5,d_4,d_3,d_2, d_1, d_0, c_3,c_2, c_1, c_0\in \mathbb{F}_{16}$.
The resulting equations become:
$$
\begin{cases}
g(0)=h(0),\\
g(1)=(z^2+z)h(1),\\
g(z)=(z^2+z)h(z),\\
g(z+1)=h(z+1),\\
g(z^2)=(z^2+1)h(z^2),\\
g(z^2+1)=(z+1)h(z^2+1),\\
g(z^2+z)=zh(z^2+z),\\
g(z^2+z+1)=z^2h(z^2+z+1).
\end{cases}
$$

\textbf{Step 1:} Precompute powers in $\mathbb{F}_{16}$

We have $z^4 = z+1$. The necessary powers are computed as follows:
\[
\begin{array}{c|ccccc}
x & x^0 & x^1 & x^2 & x^3 & x^4 \\
\hline
0 & 1 & 0 & 0 & 0 & 0 \\
1 & 1 & 1 & 1 & 1 & 1 \\
z & 1 & z & z^2 & z^3 & z+1 \\
z+1 & 1 & z+1 & z^2+1 & z^3+z^2+z+1 & z \\
z^2 & 1 & z^2 & z+1 & z^3+z^2 & z^2+1 \\
z^2+1 & 1 & z^2+1 & z & z^3+z & z^2+1 \\
z^2+z & 1 & z^2+z & z^2+z+1 & 1 & z^2+z+1 \\
z^2+z+1 & 1 & z^2+z+1 & z^2+z & 1 & z^2+z
\end{array}
\]

\textbf{Step 2:} Write equations in terms of coefficients

Now we substitute each evaluation point into $g(x)$ and $h(x)$
and equate them in the following equations:

\begin{enumerate}
    \item Equation 1: $g(0) = h(0)$ implies
    \[
    d_0 = c_0.
    \]

    \item Equation 2: $g(z+1) = h(z+1)$ gives
    \[
    d_5(z^2+z) + d_4(z) + d_3(z^3+z^2+z+1) + d_2(z^2+1) + d_1(z+1) + d_0 = c_3(z^3+z^2+z+1) + c_2(z^2+1) + c_1(z+1) + c_0.
    \]

    \item Equation 3: $g(z^2+1) = (z+1) h(z^2+1)$ gives
    \[
    d_5(z^3+z) + d_4(z) + d_3(z^3+z) + d_2(z) + d_1(z^2+1) + d_0 = (z+1)\left[c_3(z^3+z) + c_2(z) + c_1(z^2+1) + c_0\right].
    \]

    \item Equation 4: $g(z^2+z) = z h(z^2+z)$ gives
    \[
    d_5(1) + d_4(z^2+z+1) + d_3(1) + d_2(z^2+z+1) + d_1(z^2+z) + d_0 = z\left[c_3(1) + c_2(z^2+z+1) + c_1(z^2+z) + c_0\right].
    \]

    \item Equation 5: $g(z^2+z+1) = z^2 h(z^2+z+1)$ gives
    \[
    d_5(z^2+z) + d_4(z^2+z) + d_3(1) + d_2(z^2+z) + d_1(z^2+z+1) + d_0 = z^2\left[c_3(z^2+z) + c_2(z^2+z) + c_1(z^2+z+1) + c_0\right].
    \]

    \item Equation 6: $g(1) = (z^2+z) h(1)$ gives
    \[
    d_5 + d_4 + d_3 + d_2 + d_1 + d_0 = (z^2+z)(c_3 + c_2 + c_1 + c_0).
    \]

    \item Equation 7: $g(z) = (z^2+z) h(z)$ gives
    \[
    d_5(z^2+z+1) + d_4(z+1) + d_3(z^3) + d_2(z^2) + d_1(z) + d_0 = (z^2+z)\left[c_3 z^3 + c_2 z^2 + c_1 z + c_0\right].
    \]

    \item Equation 8: $g(z^2) = (z^2+1) h(z^2)$ gives
    \[
    d_5(z^2+z+1) + d_4(z^2+1) + d_3(z^3+z^2) + d_2(z+1) + d_1(z^2) + d_0 = (z^2+1)\left[c_3(z^3+z^2) + c_2(z^2+1) + c_1(z^2) + c_0\right].
    \]
\end{enumerate}

\textbf{Step 3:} Solve via Gaussian elimination

Each equation is linear in coefficients $d_i$ and $c_i$.
We can express this system in matrix form and solve the resulting system over  $\mathbb{F}_{16}$. The solution space generally turns out to have dimension 2.

The general solution can be expressed as:

\[
\begin{aligned}
h(x) &= (x + z + 1)(x+1)(ax+b), \\
g(x) &= (z^2 x^2 + z x + 1)(x + z + 1)(x+1)(ax+b),
\end{aligned}
\]

where \( a,b \in \mathbb{F}_{16} \) can be any arbitrary elements in $\mathbb{F}_{16}$.

\textbf{Step 4:} Solve the twisted polynomial and output the codeword.

The twisted polynomial $f(x)$ becomes
$$f(x)=\frac{g(x)}{h(x)}=z^2x^2+zx+1,$$
yielding the codeword $\mathbf{u}$
\begin{eqnarray*}
\mathbf{u}&=&(f(0),f(1),f(z),f(z+1),f(z^2),f(z^2+1),f(z^2+z),f(z^2+z+1))\\
&=& (1,z^2+z+1,z^2+z,0,z^2+1,z+1,z,z^2).
\end{eqnarray*}
}
\end{Example}

\begin{Remark}{\rm
In the decoding algorithm presented in  Theorems \ref{MDSGRSisodd} and \ref{MDSGRSiseven},
we only use Gaussian elimination method to find the twisted polynomial.
The algorithmic complexity of solving the homogeneous linear system to determine the coefficients of $g(x)$
and $h(x)$ is $O(n^3)$. 
In contrast, as noted in \cite{SYJL}, algorithms based
on Euclidean  algorithm exhibit a time complexity of  $O(qn)$.
Consequently, in some scenarios,
the algorithm provided in this paper may prove to be more efficient than alternatives.
}
\end{Remark}

\begin{algorithm}[H]
\caption{Decoding Algorithm for $[n,k]$-MDS TGRS Codes ($n-k$ odd)}
\begin{algorithmic}[1]

\Require
    Field $\mathbb{F}_q$; \\
    Pairwise distinct $\alpha_1,\dots,\alpha_n\in\mathbb{F}_q$; \\
    Received word $y=(y_1,\dots,y_n)\in\mathbb{F}_q$.

\Ensure
    Decoded codeword $u=(u_1,\dots,u_n)\in\mathbb{F}_q$.

\Statex \Comment{Pre-compute target degrees}
\State $s \gets (n-k-1)/2$;\quad $t \gets (n+k-1)/2$

\Statex \Comment{Build the $(n\times (s+t+2))$ matrix $A$}
\For{$j=1$ \textbf{to} $n$}
    \For{$i=0$ \textbf{to} $s$}
        \State $A_{j,i+1} \gets \alpha_j^i$
    \EndFor
    \For{$i=0$ \textbf{to} $t$}
        \State $A_{j,s+i+2} \gets -\alpha_j^i y_j$
    \EndFor
\EndFor

\Statex \Comment{Find a non-zero vector in $\ker A$}
\State $\mathbf{z}=(g_0,\dots,g_s,h_0,\dots,h_t)^T \gets$ any non-zero solution of $A\mathbf{z}=0$ by Gaussian elimination

\Statex \Comment{Reconstruct the message polynomial}
\State $g(x)\gets\sum_{i=0}^s g_i x^i$;\quad $h(x)\gets\sum_{i=0}^t h_i x^i$
\State $f(x)\gets g(x)/h(x)$

\Statex \Comment{Form the output codeword}
\For{$j=1$ \textbf{to} $n$}
    \State $u_j \gets f(\alpha_j)$
\EndFor

\State \Return $u=(u_1,\dots,u_n)$

\end{algorithmic}
\end{algorithm}

\section{A decoding algorithm for NMDS TGRS codes}
In this section, we introduce a decoding algorithm for a
class of NMDS TGRS  codes, denoted as
$\mathcal{C}_{k,1,h}(\alpha,\mathbf{1},\eta)$, where $\mathbf{1}=(1,1,\cdots,1)$. The following lemma is often employed to determine whether a TGRS code is NMDS.

\begin{lem}(\cite[Lemma 2.6]{HY})\label{nmdscondition}
Let $\mathcal{C}_{k,1,h}(\alpha,\mathbf{1},\eta)$ be defined as above.
Then the TGRS code $\mathcal{C}_{k,1,h}(\alpha,\mathbf{1},\eta)$ is NMDS if and only if
$$
\eta\sum_{i\in I}\alpha_i= -1~~\hbox{for some}~ I\subset\{1,2,\cdots,n\}
~\hbox{with}~|I|=k.
$$
\end{lem}

The following theorem follows a similar concept as Theorem \ref{MDSGRSisodd} and provides a decoding algorithm for NMDS TGRS codes.

\begin{Theorem}\label{NMDSGRS}
Let $\mathcal{C}_{k,1,h}(\alpha,\mathbf{1},\eta)$ be defined as above.
Then there exists a decoding algorithm for the $k$-dimensional NMDS TGRS code $\mathcal{C}_{k,1,h}(\alpha,\mathbf{1},\eta)$ of length $n$,
which can correct  up to $\lfloor \frac{n-k-1}{2}\rfloor$ errors and executes in a number of operations that is polynomial in $n$.
\end{Theorem}
\begin{proof}
Suppose that we have received the vector $(y_1,y_2,\cdots,y_n)$. We want to find the $(k,1,h,\eta)$-twisted polynomial $f(x)\in \mathcal{V}_{(k,1,h,\eta)}$ with degree at most $k$
such that
\begin{equation}\label{error vector41}
(y_1,y_2,\cdots,y_n)=(f(\alpha_1),f(\alpha_2),\cdots,f(\alpha_n))+\mathbf{e},
\end{equation}
where $\mathbf{e}$ is the error vector with weight at most $\lfloor \frac{n-k-1}{2}\rfloor$.
Observe that  $\mathcal{C}_{k,1,h}(\alpha,\mathbf{1},\eta)$ is an NMDS code,
we have
$$\frac{n-k-1}{2}=\frac{d-1}{2}.$$

Let $h(x)$ be an arbitrary polynomial of degree at most $\lfloor\frac{n-k-1}{2}\rfloor$ and let $g(x)$ be an arbitrary polynomial of degree at most $k+\lceil\frac{n-k-1}{2}\rceil$.
We determine the coefficients of $g(x)$ and $h(x)$ by solving the system of $n$ equations:
$$g(\alpha_j)-h(\alpha_j)y_j=0,\mbox{for} ~j=1,2,\cdots,n.$$
This homogeneous linear system has
$$\lfloor\frac{n-k-1}{2}\rfloor+1+k+\lceil\frac{n-k-1}{2}\rceil+1=n+1$$
unknowns (the coefficients of $g(x)$ and $h(x)$) and $n$ equations. Hence, we can find a non-trival solution for $h(x)$ and
$g(x)$) in a number of operations that is polynomial in $n$ using Guassian elimination.

By assumption (\ref{error vector41}), there is a polynomial $f(x)$ of degree at most $k$
such that $y_j=f(\alpha_j)$ for at least $n-\lfloor\frac{n-k-1}{2}\rfloor$ values of $j$.
Thus we let $J=\{j\,|\,y_j=f(\alpha_j), 1\leq j\leq n\}$ and then the number $|J|$ of elements in the set $J$  has the low bound as follows:
$$|J|\geq n-\lfloor\frac{n-k-1}{2}\rfloor.$$

On the other hand, we have
$$g(\alpha_j)-h(\alpha_j)f(\alpha_j)=g(\alpha_j)-h(\alpha_j)y_j=0~
 \hbox{for any}~j\in J.$$
That is to say, for these values of $j\in J$, $\alpha_j$ is a zero of $g(x)-f(x)h(x)$.
Hence, this polynomial $g(x)-f(x)h(x)$ has at least $|J|\geq n-\lfloor\frac{n-k-1}{2}\rfloor$ distinct zeros in $\mathbb{F}_q$.

Now we consider the degree of the polynomial $\varphi(x)=g(x)-h(x)f(x)$. Suppose that $\varphi(x)\neq 0$. It follows that
\begin{eqnarray*}
\deg(\varphi(x))&=& \max\big\{\deg(g(x)), \deg(h(x)f(x))\big\}\\
&\leq & \max\big\{k+\lceil\frac{n-k-1}{2}\rceil, \lfloor\frac{n-k-1}{2}\rfloor+k\big\}\\
&\leq & \max\big\{k+(n-k-1)-\lfloor\frac{n-k-1}{2}\rfloor, \lfloor\frac{n-k-1}{2}\rfloor+k\big\}\\
&\leq & \max\big\{n-1-\lfloor\frac{n-k-1}{2}\rfloor, \lfloor\frac{n-k-1}{2}\rfloor+k\big\}\\
&= & n-1-\lfloor\frac{n-k-1}{2}\rfloor=k+\lceil\frac{n-k-1}{2}\rceil.
\end{eqnarray*}
Given that a polynomial cannot have more roots than its degree, we conclude that
$$n-\lfloor\frac{n-k-1}{2}\rfloor\leq k+\lceil\frac{n-k-1}{2}\rceil.$$
However, we find that
\begin{eqnarray*}
k+\lceil\frac{n-k-1}{2}\rceil& = &k+(n-k-1-\lfloor\frac{n-k-1}{2}\rfloor)\\
& = & n-1-\lfloor\frac{n-k-1}{2}\rfloor \\
&<&n-\lfloor\frac{n-k-1}{2}\rfloor.
\end{eqnarray*}
This gets a contradiction, implying $\varphi(x)$ is identically zero. Therefore, $h(x)$ divides $g(x)$ and the quotient is $f(x)$.
\end{proof}

In the following, we give an example to illustrate the above result.

\begin{Example}{\rm
Suppose that $n=q=7,k=2$ and $\eta=2\in \mathbb{F}_7$. Then $\mathbb{F}_7=\{\alpha_1,\alpha_2,\alpha_3,\alpha_4,\alpha_5,\alpha_6,\alpha_7\}$,
where $\alpha_i=i-1, i=1,2,\cdots,7$.
Since $\eta(\alpha_1+\alpha_4)=2(0+3)=6=-1$, from Lemma \ref{nmdscondition} we get that the TGRS code $\mathcal{C}_{2,1,0}(\alpha,\mathbf{1},2)$ is NMDS,
where
$$
\mathcal{C}_{2,1,0}(\alpha,\mathbf{1},2)=
\big\{(f(0),f(1),f(2),f(3),f(4),f(5),f(6))\,\big|\,f(x)=a_0+a_1x+2a_0x^2
 ~\hbox{for any}~a_0,a_1\in \mathbb{F}_7\big\}.
$$
Suppose that we have sent a codeword $\mathbf{u}$ of the $2$-dimensional NMDS GRS code $\mathcal{C}_{2,1,0}(\alpha,\mathbf{1},2)$ over $\mathbb{F}_7=\{0,1,2,3,4,5,6\}$
(order  the elements of $\mathbb{F}_7$ in this order) and that
we have received
$$\mathbf{y}=(y_1,y_2,y_3,y_4,y_5,y_6,y_7)=(1,1,0,0,3,3,0).$$
According to the algorithm in the proof of Theorem \ref{NMDSGRS},
we should find a polynomial $g(x)$ of degree at most $4$ and a polynomial $h(x)$ of degree at most $2$, such that
$$g(\alpha_j)=h(\alpha_j)y_j $$
for $j=0,1,2,3,4,5,6$,  where $\alpha_j$ is the $j$-th element of $\mathbb{F}_7$.

The equations are
$$g(0)=h(0), g(1)=h(1), g(2)=g(3)=g(6)=0, g(4)=3h(4), g(5)=3h(5).$$

From this we deduce that
$$g(x)=(x-2)(x-3)(x-6)(g_1x+g_0)$$
and
$$h(x)=h_2x^2+h_1x+h_0,$$
for some $h_2, h_1, h_0, g_1, g_0\in \mathbb{F}_7$, which are solutions of the system
$$
\begin{cases}
g(0)=h(0),\\
g(1)=h(1),\\
g(4)=3h(4),\\
g(5)=3h(5),
\end{cases}
$$
i.e.,
$$
\begin{cases}
-g_0=h_0,\\
-3(g_1+g_0)=h_2+h_1+h_0,\\
-2g_1-4g_0=-h_2-2h_1+3h_0,\\
5g_1+g_0=-2h_2+h_1+3h_0,
\end{cases}
$$
which implies that all the solutions of the system of equations are
$$
\begin{cases}
g_0=-2h_2,\\
g_1=2h_2,\\
h_0=2h_2,\\
h_1=4h_2,
\end{cases}
$$
where $h_2\in \mathbb{F}_7$.
Thus we have
$$
\begin{cases}
g(x)=(x-2)(x-3)(x-6)[2h_2x-2h_2],\\
h(x)=h_2x^2+4h_2x+2h_2.
\end{cases}
$$
We can readily check that $h(1)=0$ and then we obtain
\begin{eqnarray*}
f(x)=\frac{g(x)}{h(x)}&=&\frac{(x-2)(x-3)(x-6)(2x-2)}{(x-1)(x-2)}\\
&=&2(x-3)(x-6)=1+3x+2x^2\\
&=&a_0+a_1x+2a_0x^2.
\end{eqnarray*}
Evaluating the polynomial $f(x)$, we deduce that
$$\mathbf{u}=\big(f(0),f(1),f(2),f(3),f(4),f(5),f(6)\big)=(1,6,1,0,3,3,0).$$
}
\end{Example}

\section{Comparison with the known results}
Recall that the TGRS code $\mathcal{C}_{k,1,h}(\alpha,\mathbf{1},\eta)$ is
defined as
$$\mathcal{C}_{k,1,h}(\alpha,\mathbf{1},\eta)=
\Big\{\big(f(\alpha_1),f(\alpha_2),\cdots,f(\alpha_n)\big)\,\Big|\,f(x)\in \mathcal{V}_{(k,1,h,\eta)}\Big\},$$
where  $f(x)\in \mathcal{V}_{(k,1,h,\eta)}$ means that $f(x)$ is of the form
$$f(x)=\sum_{i=0}^{k-1}a_ix^i+\eta a_hx^{k}.$$
According to Lemma \ref{mdscondition}, the TGRS code $\mathcal{C}_{k,1,h}(\alpha,\mathbf{1},\eta)$ is MDS if and only if
$$
\eta\sum_{i\in I}\alpha_i\neq -1~\hbox{for any}~I\subset\{1,2,\cdots,n\}
~\hbox{with}~|I|=k.
$$
Assume that for $i=1,2,\cdots,n$,
$$u_i=\prod_{j=1,j\neq i}^n(\alpha_i-\alpha_j)
~\hbox{and}~\lambda=\sum_{i=1}^n\alpha_i.$$

In \cite{SYJL}, Sun et al. presented the key equation for decoding MDS TGRS codes  along with the corresponding decoding algorithm. The so-called key equation seeks to solve the following relationship:
$$S(x)\sigma(x)\equiv \tau(x)~(\textrm{mod} ~G(x))$$
where  $S(x)$ and $G(x)$ are given polynomials,
the degree of $\sigma(x)$ matches the number of errors,
and the degree of $\tau(x)$ is less than or equal to that of $\sigma(x)$.
Specifically, Sun et al. in \cite{SYJL} focused on two types of
 $[n,k,n-k+1]$ MDS TGRS codes
based on the extended Euclidean algorithm, particularly for cases when $n-k$ is even. Their respective parity-check matrices are given by:
$$H_1=
\begin{pmatrix}
v_1(1+\eta\alpha_1^{n-k}) & v_2(1+\eta\alpha_2^{n-k}) &\cdots & v_n(1+\eta\alpha_n^{n-k})\\
v_1\alpha_1 & v_2\alpha_2 & \cdots & v_n\alpha_n\\
\vdots & \vdots & \ddots & \vdots \\
v_1\alpha_1^{n-k-2} & v_2\alpha_2^{n-k-2} & \cdots & v_n\alpha_n^{n-k-2}\\
v_1\alpha_1^{n-k-1} & v_2\alpha_2^{n-k-1} & \cdots & v_n\alpha_n^{n-k-1}
\end{pmatrix}
$$
and
$$H_2=
\begin{pmatrix}
v_1 & v_2 &\cdots & v_n\\
v_1\alpha_1 & v_2\alpha_2 & \cdots & v_n\alpha_n\\
\vdots & \vdots & \ddots & \vdots \\
v_1\alpha_1^{n-k-2} & v_2\alpha_2^{n-k-2} & \cdots & v_n\alpha_n^{n-k-2}\\
v_1(\alpha_1^{n-k-1}+\eta\alpha_1^{n-k}) & v_2(\alpha_2^{n-k-1}+\eta\alpha_1^{n-k}) & \cdots & v_n(\alpha_n^{n-k-1}+\eta\alpha_1^{n-k})
\end{pmatrix}.
$$
These types of MDS TGRS codes are characterized
by generator matrices of the forms:
$$G_1=
\begin{pmatrix}
w_1 & w_2 &\cdots & w_n\\
w_1\alpha_1 & w_2\alpha_2 & \cdots & w_n\alpha_n\\
\vdots & \vdots & \ddots & \vdots \\
w_1\alpha_1^{k-2} & w_2\alpha_2^{k-2} & \cdots & w_n\alpha_n^{k-2}\\
w_1(\alpha_1^{k-1}+\mu_1\alpha_1^{-1}) & w_2(\alpha_2^{k-1}+\mu_1\alpha_1^{-1}) & \cdots & w_n(\alpha_n^{k-1}+\mu_1\alpha_1^{-1})
\end{pmatrix}
$$
and
$$G_2=
\begin{pmatrix}
w_1 & w_2 &\cdots & w_n\\
w_1\alpha_1 & w_2\alpha_2 & \cdots & w_n\alpha_n\\
\vdots & \vdots & \ddots& \vdots \\
w_1\alpha_1^{k-2} & w_2\alpha_2^{k-2} & \cdots & w_n\alpha_n^{k-2}\\
w_1(\mu_2\alpha_1^{k-1}+\alpha_1^{k}) & w_2(\mu_2\alpha_2^{k-1}+\alpha_1^{k}) & \cdots & w_n(\mu_2\alpha_n^{k-1}+\alpha_1^{k})
\end{pmatrix},
$$
where
$$\mu_1=-\frac{\eta\sum_{i=1}^nu_i\alpha_i^{n-1}
+\sum_{i=1}^nu_i\alpha_i^{k-1}}{\sum_{i=1}^nu_i\alpha_i^{-1}},~
\mu_2=-\frac{\sum_{i=1}^nu_i\alpha_i^{n-1}+
\eta\sum_{i=1}^nu_i\alpha_i^{n}}{\eta\sum_{i=1}^nu_i\alpha_i^{n-1}},~
w_i=\frac{u_i}{v_i},
$$
and $n-k$ is even.

By comparing our results with those in \cite{SYJL}, we note the following significant aspects:
\begin{itemize}
\item[{\rm(1)}]
 We have presented a decoding algorithm for MDS
 TGRS codes over finite fields, featuring parameters
 of code lengths and dimensions that are more flexible than those provided in \cite{SYJL}. Notably, \cite{SYJL} focuses solely on the case when $n-k$ is even, whereas we have addressed both cases of $n-k$ being odd and even, as shown in Theorems \ref{MDSGRSisodd} and \ref{MDSGRSiseven} respectively. Furthermore, the hook $h$ of each MDS TGRS code in this paper is arbitrary, differing from the restricted cases when $h = 0$ or $k-1$ in \cite{SYJL}.

\item[{\rm(2)}]
Our findings also encapsulate certain situations detailed in \cite{SYJL}.
 For instance, if we let
  $\lambda\neq 0$ and $\eta\neq -\lambda^{-1}$,
   then according to \cite[Theorem 2.4]{HY}, the parity-check matrix of the MDS TGRS code $\mathcal{C}_{k,1,h}(\alpha,\mathbf{1},\eta)$
    is represented as:
$$
\begin{pmatrix}
u_1 & u_2 & \cdots & u_n\\
u_1\alpha_1 & u_2\alpha_2 & \cdots & u_n\alpha_n\\
\vdots & \vdots & \ddots& \vdots\\
u_1\alpha_1^{n-k-2} & u_2\alpha_2^{n-k-2} & \cdots & u_n\alpha_n^{n-k-2}\\
u_1(\alpha_1^{n-k-1}-\frac{\eta}{1+\lambda\eta}\alpha_1^{n-k}) & u_2(\alpha_2^{n-k-1}-\frac{\eta}{1+\lambda\eta}\alpha_2^{n-k})  & \cdots & u_n(\alpha_n^{n-k-1}-\frac{\eta}{1+\lambda\eta}\alpha_n^{n-k})
\end{pmatrix}.
$$
Given that
$$
\eta\sum_{i\in I}\alpha_i\neq -1~\hbox{for any} ~I\subset\{1,2,\cdots,n\}
~\hbox{with}~ |I|=k,
$$
it follows that
 $-\frac{\eta}{1+\lambda\eta}$ satisfies
$$
\frac{\eta}{1+\lambda\eta}\sum_{i\in J}\alpha_i\neq 1~\hbox{for any}~ ~J\subset\{1,2,\cdots,n\}~\hbox{with}~|J|=n-k.
$$
Thus, in this case, the MDS TGRS code $\mathcal{C}_{k,1,h}(\alpha,\mathbf{1},\eta)$ coincides with the TGRS code $\mathcal{C}_2$ of \cite{SYJL}.

\item[{\rm(3)}]  The decoding algorithm provided here is not only  applicable to TGRS codes, but also to twisted Goppa codes,
which was introduced in \cite{SY}.
Since the twisted Goppa codes are the subfield subcode of   TGRS codes, there exists a polynomial time decoding algorithm which
corrects the errors, whose proof is similar to that of
Theorem \ref{MDSGRSisodd} or \ref{MDSGRSiseven}.
This kind of decoding algorithm of twisted Goppa codes is also different from   the ones provided in \cite{SY} and \cite{SYJL}.
In \cite{SY}, Sui and Yue introduced twisted Goppa codes
as subfield subcodes of twisted Reed-Solomon (TRS)
codes and developed an efficient decoding algorithm based
on the extended Euclidean algorithm.
They also explored quasi-cyclic and cyclic structures of these codes to reduce public key sizes in the Niederreiter cryptosystem.
However, their decoding algorithm was limited to correcting up to $\lfloor\frac{n-k-1}{2}\rfloor$ errors when
the Goppa polynomial has even degree $n-k$,
falling short of the theoretical maximum $\lfloor\frac{n-k+1}{2}\rfloor$ errors.

\item[{\rm(4)}]  In all the algorithms outlined in Theorems \ref{MDSGRSisodd}, \ref{MDSGRSiseven}, and \ref{NMDSGRS},
we rely solely on Gaussian elimination to find the
twisted polynomial and subsequently derive the corrected codewords.
As a result, the time complexity of each algorithm is $O(n^3)$.
In contrast, the algorithms in \cite{SYJL},
based on Euclidean algorithm, exhibit a time complexity of $O(qn)$.
Thus, in certain scenarios, the algorithms provided in
this paper can demonstrate greater efficiency.
\end{itemize}
\section{Conclusion and future work}

In this paper, we have introduced a new decoding algorithm for
MDS TGRS codes with parameters $[n,k,n-k+1]$,
which relies exclusively on  \textbf{Gaussian elimination}.
The algorithm successfully corrects up to:
\begin{itemize}
  \item $\displaystyle \Bigl\lfloor\frac{n-k}{2}\Bigr\rfloor$ errors when $n-k$ is odd, and
  \item $\displaystyle \Bigl\lfloor\frac{n-k}{2}\Bigr\rfloor-1$ errors when $n-k$ is even.
\end{itemize}
Moreover, it operates with a time complexity of
 $O(n^{3})$ operations.
Our method distinguishes itself from previous approaches
based on Euclidean algorithm by eliminating the need for polynomial divisions and the extended GCD step. Notably, it also addresses the previously unexplored odd-$(n-k)$ case highlighted in the literature~\cite{SYJL}.
Additionally, we demonstrated that this framework
is applicable to NMDS TGRS codes with parameters $[n,k,n-k]$,
thereby allowing for the correction of
$\lfloor\frac{n-k-1}{2}\rfloor$
errors while maintaining polynomial-time complexity.

Possible
future work will focus on leveraging sparsity or banded structures to enhance practical running time, as well as extending the method to higher-order
TGRS codes and other algebraic-geometric codes to
evaluate its effectiveness and robustness across a
 broader range of parameters.




\end{document}